\documentclass[10pt,journal,compsoc]{IEEEtran}

\usepackage[english]{babel}
\usepackage{amsmath,amssymb}
\usepackage{times}
\usepackage{relsize}
\usepackage{graphicx}
\usepackage{url}
\usepackage{booktabs}
\usepackage{multirow}
\usepackage{mparhack}
\usepackage{subfigure}
\usepackage{amsthm}
\usepackage{amsfonts}

\usepackage[noend]{algorithmic}
\usepackage[linesnumbered,ruled,vlined]{algorithm2e}

\usepackage{array}
\usepackage{stfloats}
\usepackage{cite}
\usepackage{eqparbox}
 \usepackage{amsmath}
\usepackage{mdwmath}
\usepackage{balance}
\usepackage{epsfig}
\usepackage{xcolor}
\usepackage{mathtools}
\usepackage{bm}
\usepackage{amsfonts}
\usepackage[all]{xy}
\usepackage{etoolbox}
\usepackage{graphicx}
\usepackage{stfloats}
\DeclareMathAlphabet{\mathantt}{OT1}{antt}{li}{it}
\DeclareMathAlphabet{\mathpzc}{OT1}{pzc}{m}{it}

\usepackage{accents}

\newtheorem{theorem}{Theorem}

\newtheorem{lemma}[theorem]{Lemma}

\DeclareFontFamily{OT1}{pzc}{}
\DeclareFontShape{OT1}{pzc}{m}{it}%
  {<-> s * [1.1] pzcmi7t}{}
\DeclareMathAlphabet{\mathpzc}{OT1}{pzc}%
                     {m}{it}

\DeclareMathOperator{\argmax}{\arg\max}

\makeatletter
\patchcmd{\@maketitle}
  {\addvspace{0.75\baselineskip}\egroup}
  {\addvspace{-1.45\baselineskip}\egroup}
  {}
  {}

\begin{document}

\title{Task Offloading Optimization in Mobile Edge Computing under Uncertain Processing Cycles and Intermittent Communications}

\author{Tao Deng,
 Zhanwei Yu,
 and Di Yuan\\
 \IEEEcompsocitemizethanks{\IEEEcompsocthanksitem T.\ Deng is with the School of Computer Science and Technology, Soochow University, Suzhou, Jiangsu 215006, China.\protect\\

E-mail: dengtao@suda.edu.cn
\IEEEcompsocthanksitem Z.\ Yu and D.\ Yuan are with the Department of Information Technology, Uppsala University, 751 05 Uppsala, Sweden.\protect\\
 E-mail: zhanwei.yu; di.yuan@it.uu.se
}

}

\IEEEtitleabstractindextext{%
\begin{abstract}
Mobile edge computing (MEC) has been regarded as a promising approach to deal with explosive computation requirements by enabling cloud computing capabilities at the edge of networks.
Existing models of MEC impose some strong assumptions on the known processing cycles and unintermittent communications.
However, practical MEC systems are constrained by various uncertainties and intermittent communications, rendering these assumptions impractical.
In view of this, we investigate how to schedule task offloading in MEC systems with uncertainties.
First, we derive a closed-form expression of the average offloading success probability in a device-to-device (D2D) assisted MEC system with uncertain computation processing cycles and intermittent communications.
Then, we formulate a task offloading maximization problem (TOMP), and prove that the problem is NP-hard.
For problem solving, if the problem instance exhibits a symmetric structure, we propose a task scheduling algorithm based on dynamic programming (TSDP).
 By solving this problem instance, we derive a bound to benchmark sub-optimal algorithm.
 For general scenarios, by reformulating the problem, we propose a repeated matching algorithm (RMA).
Finally, in performance evaluations, we validate the accuracy of the closed-form expression of the average offloading success probability by Monte Carlo simulations, as well as the effectiveness of the proposed algorithms.
\end{abstract}

}

\maketitle

\IEEEdisplaynontitleabstractindextext

\IEEEpeerreviewmaketitle

\IEEEraisesectionheading{\section{Introduction}\label{sec:introduction}}
Mobile applications running on smart devices bring explosive computation requirements.
Sometimes the required resources by the applications exceed the computational capacity of the devices.
With cloud computing, mobile devices can offload their computing tasks to the cloud via wireless networks.
However, this approach will add the burden of backhaul links and transmission cost.

Mobile edge computing (MEC) has been regarded as a promising approach to
deal with computation offloading \cite{Bon2012Fog}.
By the approach, mobile devices can offload their computing tasks to edge nodes configured with computing resources, e.g., base stations (BSs), so as to reduce the computational delay, energy consumption, and so on.
In addition to offloading the tasks to BSs that is subject to a limited spectrum, device-to-device (D2D) assisted MEC is an effective solution to improve the efficiency of MEC by exploiting the nearby devices' spare computing resources \cite{Xu2018Sur,Waqas2020A}.
In order to realize this paradigm, existing models encounter two challenges.
First, due to the device mobility, the communication connection between devices is intermittent.
However, most of existing models do not consider this aspect.
In addition, most of existing models assume that the number of processing cycles of a task is known in advance.
This assumption is too ideal because for some tasks the number of processing cycles may be dependent on input size, but the relation still can be complex and only some statistical data and properties are available in a practical MEC system \cite{BLiang2019In}.
Thus, it is hard to predict the number of processing cycles.
Therefore, existing models need to be enhanced to cope with both the uncertain processing cycles and intermittent communications.

We investigate how to schedule task offloading in a D2D-assisted MEC system with uncertain computation processing cycles and intermittent communications.
Our objective is to maximize the average offloading success probability over all the tasks.
The main contributions of this paper are as follows.
\begin{itemize}
\item
First, a closed-form expression of the average offloading success probability is derived.
Based on this expression, a task offloading maximization problem (TOMP) is formulated.

\item
Second, the hardness of the problem based on a reduction from the Knapsack problem is proved.

\item
Moreover, to solve the problem, a task scheduling algorithm based on dynamic programming (TSDP) is proposed for the scenario where the transmission capacities and residual time for all MEC nodes are uniform.
By solving the uniform scenario, a bound on the optimization problem is derived, which can be used to benchmark any sub-optimal algorithm.
\item
 For general scenarios of TOMP, by reformulating the problem, a repeated matching algorithm (RMA) is proposed to solve the problem.
In the algorithm, a series of matching problems are solved, and the input matrix is updated in every iteration, until the solution can not be improved.
\item
Finally, the effectivenesses of the proposed algorithms is validated through our performance evaluations,
that is, the accuracy of the closed-form expression of the offloading success probability, small gaps (less than $0.55\%$) between the closed-form expression and Monte Carlo simulation results, etc. It is shown that in the uniform scenario, TSDP is effective as its solution is close to the bound overall; for general scenarios, RMA outperforms other algorithms.
\end{itemize}

The remainder of this paper is organized as follows.
Section II presents the related works.
Section III presents two preparatory propositions.
Section IV introduces the system scenario, computation processing model, mobility model, and task offloading model.
Section V first formulates our optimization problem, and then proves that the problem is NP-hard.
Section VI introduces the TSDP algorithm for the uniform scenario, and the RMA algorithm for general scenarios.
Section VII presents performance evaluations.
Finally, Section VIII concludes this paper.
\section{Related work}
The existing works on D2D-assisted MEC can be divided into two categories.

The works in \cite{Hamdi2022,Zeng2022Delay,Peng2021D2d,MChen2022Sig,Zhou2021Inc,Pu2016D2d,Li2021An} investigate task offloading in static environments where they do not consider mobility.
In \cite{Hamdi2022}, the work models a joint task scheduling and power allocation optimization problem.
For problem solving, the authors first decouple the problem into a power allocation problem and an offloading assignment problem, and then propose the conjugate gradient algorithm and the Hungarian algorithm to solve the corresponding problems, respectively.
In \cite{Zeng2022Delay}, the authors propose a two-step algorithm to jointly optimize service sharing, computation offloading, and bandwidth allocation.
In \cite{Peng2021D2d}, the work models a jointly network-wide delay and power consumption optimization problem, and proposes an online resource scheduling algorithm.
The works in \cite{MChen2022Sig,Zhou2021Inc,Pu2016D2d,Li2021An} analyze the performance of D2D-assisted MEC by proposing some incentive mechanisms.
In \cite{MChen2022Sig} and \cite{Zhou2021Inc}, the works propose contract-based incentive mechanisms to motivate local MEC nodes to take part in D2D computation and content offloading.
In \cite{Pu2016D2d}, the work formulates a time-average energy consumption minimization problem subject to some incentive constraints.
The work in \cite{Li2021An} models an incentive-aware optimization problem based on a novel utility function, and develops a price-based algorithm.

The works in \cite{Zhou2020Fre,Han2020Opp,Mu2020Online,Qin2021Sparse,Wang2019Mobil,Saleem2021Mobil,Wang2018Us,Wang2014Mobi,Chen2018Oppor,Ahani2019BS} investigate task offloading in mobile environments.
In \cite{Zhou2020Fre}, the authors model a utility maximization problem considering the transmission overhead and the freshness of contents, and propose two algorithms for problem solving.
In \cite{Han2020Opp}, the work proposes a Markov decision process framework to analyze the average delay and cost in vehicular networks.
In \cite{Mu2020Online}, the authors investigate the case of one requester generating tasks and one helper with computing resources.
 They model an energy minimization scheduling problem that takes the arrival processes of tasks, the helper's resources, and channel conditions into consideration.
 For problem solving, they design an online algorithm to derive the optimum scheduling policy.
In \cite{Qin2021Sparse}, the authors maximize data offloading ratio subject to the delay constraint, and develop a greedy algorithm to solve the problem.
In \cite{Wang2019Mobil}, the authors investigate a three-layer MEC architecture, where a sojourn time model is used to characterize users' mobility and the sojourn time of users is assumed to follow an exponential distribution.
They divide the problem into a task scheduling problem and a resource allocation problem, and propose algorithms to solve them.
In \cite{Saleem2021Mobil}, the authors formulate a jointly task scheduling and power allocation optimization problem, which is a mixed-integer non-linear programming (MINLP) problem.
They first propose an algorithm based on genetic algorithm to solve the MINLP problem, and then propose a low complexity mobility-aware task scheduling algorithm.
In \cite{Wang2018Us}, the authors formulate a task allocation problem as a constant satisfaction problem that takes into account mobility, task properties, and network constraints.
For problem solving, they propose a lightweight heuristic algorithm.
In \cite{Wang2014Mobi}, the authors consider one task.
They divide the task into sub-tasks, and investigate how to offload the sub-tasks to other mobile computing nodes.
In \cite{Chen2018Oppor}, the authors address a delay-cost tradeoff optimization problem in
opportunistic task scheduling scenarios.
In \cite{Ahani2019BS}, the authors develop a BS-assisted computation offloading scheme.
All these works assume that the processing cycles of tasks are known in advance.
This assumption is too ideal because for some tasks the number of processing cycles may be dependent on input size, but the relation still can be complex and only some statistical data and properties are available.
This uncertainty brings great challenges to the modelling and optimization of task scheduling.

A few works investigate uncertainty in processing cycles \cite{Hou2021Task,Li2022Non}.
In \cite{Hou2021Task}, the authors formulate a task offloading problem and propose an energy-minimized solution with probabilistic deadline guarantee (EPD).
Simulation results show that EPD provides significant gains in energy saving.
In \cite{Li2022Non}, the author considers both offline and online non-clairvoyant task offloading, and proposes a non-clairvoyant task offloading algorithm for offline task offloading and a randomised online task offloading algorithm for online task offloading.
However, the works in \cite{Hou2021Task} and \cite{Li2022Non} are different from ours.
In our scenario, in addition to the uncertain processing cycles, the communication connection between the devices is intermittent.
Thus, the system modeling and optimization of our scenario is more challenge than that of \cite{Hou2021Task} and \cite{Li2022Non}.

\section{Preliminaries}

\textit{Proposition 1.} Suppose that $a$ and $b$ are two continuous positive random variables with probability density function (pdf) $f_a(t)$ and $f_b(t)$, respectively.
The Laplace transform of $f_a(t)$ and $f_b(t)$ are denoted by $f^*_a(s)$ and $f^*_b(s)$, respectively.
The set of poles of $f^*_a(s)$ is denoted by $\Omega_a$.
It follows from \cite{Fang1998Channel} that
\begin{equation}
\begin{aligned}
p(a>b)=-\underset{q \in \Omega_a}\sum\underset{s=q}{\text{Res}}\frac{f^*_b(s)}{s}f_a^*(-s),
\end{aligned}
\label{XY}
\end{equation}
where $\underset{s=q}{\text{Res}}$ represents the residue at pole $s = q$.

\textit{Proposition 2}. Suppose that $a_0, a_1, a_2, \dots, a_k$ are continuous positive  independent and identically distributed (i.i.d) random variables with pdf $f_a(t)$.
 Suppose that $b$ is a continuous positive random variable with pdf $f_b(t)$.
It follows from \cite{Fang1998Channel} that
\begin{equation}
\begin{aligned}
&p(a_0+a_1+a_2+\dots+a_k<b)\\
&=-\underset{q \in \Omega_b}\sum\underset{s=q}{\text{Res}}\frac{f_{a_0}^*(s){f^*(s)}^k}{s}f^*_b(-s).
\end{aligned}
\label{XXY}
\end{equation}

\section{System Model}
\subsection{System Scenario}
\begin{figure}[ht]
\centering
\includegraphics[scale=0.65]{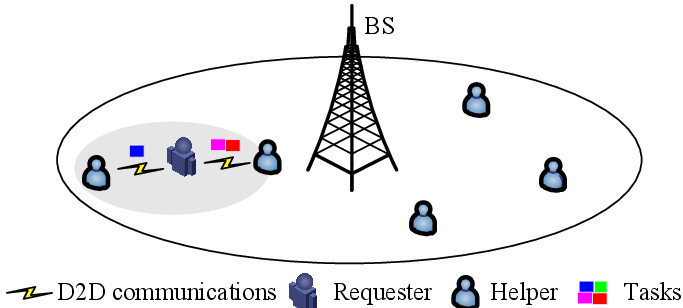}
\begin{center}
\caption{System scenario.}
\label{SystemScenario}
\end{center}
\end{figure}
We consider a D2D-assisted MEC system with one node\footnote[1]{Our model and algorithm can be applied also to the scenario of multiple requesters.} that has the task offloading requirements, referred to as \emph{requester}, and $H$ nodes with spare computation resources that will help offload the requester's tasks, referred to as \emph{helpers}, as shown in Fig.~\ref{SystemScenario}.
The set of helpers is denoted by $\mathcal{H}$, $\mathcal{H}=\{1,2,\dots,H\}$.
The requester produces $R$ tasks.
The set of tasks is denoted by $\mathcal{R}$, $\mathcal{R}=\{1,2,\dots,R\}$.
Denote by $l_i$ the size of task $i$.
We define our optimization matrix, denoted by $\bm{x}$,
\begin{equation}
\begin{aligned}
\bm{x}=\{x_{ij}, i\in \mathcal{R}~\text{and}~j\in \mathcal{H} \},
\end{aligned}
\end{equation}
where $x_{ij} \in \{0, 1\}$, which is one if and only if task $i$ is offloaded to helper $j$.
We consider that there is a preparation time duration in which the requester transmits the tasks to helpers by D2D communications.
Denote by $E_j$ the transmission capacity from the requester to helper $j$ in the preparation time duration.
It follows that for helper $j$,
\begin{equation}
\begin{aligned}
\sum_{i=1}^{R}  x_{ij}l_i \leq E_j.
\label{CapConst}
\end{aligned}
\end{equation}
For any task $i\in \mathcal{R}$, it follows that
\begin{equation}
\begin{aligned}
\sum_{j=1}^{H}  x_{ij} \leq 1.
\label{NumCons}
\end{aligned}
\end{equation}
\subsection{Computation Processing Model}
Denote by $t_{ij}$ the processing cycles for helper $j$ to completely compute task $i$.
We consider that $t_{ij}$ is a random variable.
Denote by $h_{ij}(t)$ the pdf of $t_{ij}$ with mean $\textbf{E}(t_{ij})=\xi_{ij}$ and variance $\textbf{Var}(t_{ij})=\sigma^2_{ij}$.

\subsection{Mobility Model}
The requester can communicate with any helper when they move into the range of D2D communication.
Due to the nodes' movement, the communication process is intermittent.
Denote by $t^c_{jk}$ the $k$-th communication time period (CTP) of helper $j$.
The offloading process happens only in the first CTP.
The $t^c_{jk}$'s are i.i.d. random
variables with mean $\textbf{E}(t^c_{jk})=\mu_j$ and variance $\textbf{Var}(t^c_{jk})=\sigma^2_{c_j}$, $k=2,3,\dots,+\infty$.
Denote by $f_j(t)$ the pdf of $t^c_{jk}$.
Denote by $t^r_{j1}$ the residual CTP from the moment that a task is offloaded to helper $j$ to the moment that the first CTP ends.
Denote by $f_{r_j}(t)$ the pdf of $t^r_{j1}$.
It follows from the residual life theorem \cite{cox1962renewal} that,
\begin{equation}
\begin{aligned}
f_{r_j}(t)=\mu_j \int_{x=t}^{+\infty} f_j(x) dx.
\end{aligned}
\end{equation}
Denote by $t^s_{jk}$ the $k$-th inter-communication time period (ICTP) between the requester and helper $j$ from the moment that the $k$-th CTP ends to the moment that the $(k+1)$-th CTP starts, $k=1,2,\dots,+\infty$.
The $t^s_{jk}$'s are i.i.d. random
variables with mean $\textbf{E}(t^s_{jk})=\gamma_j$ and variance $\textbf{Var}(t^s_{jk})=\sigma^2_{s_j}$.
Denote by $w_j(t)$ the pdf of $t^s_{jk}$.


\subsection{Offloading Success Model}
Suppose that task $i$ is offloaded to helper $j$.
The offloading successfully completes if the computational result of the task can be returned to the requester by D2D communications by the end of $t_{ij}$.
That is, the requester and the helper must be in a CTP period by the end of $t_{ij}$.
Otherwise, the helper has to send the computational result to the BS, and then the BS forwards the result to the requester.
But, this occupies the wireless resources, and increases transmission cost.
Therefore, we consider that fully D2D is counted as success, and delivering via the BS is not a success.
Fig. \ref{OffloadingProcess1} gives an example of two tasks and two helpers.
In this figure, tasks $1$ and $2$ are offloaded to helpers $1$ and $2$, respectively.
 By the end of $t_{11}$ and $t_{22}$, helper $1$ successfully completes this offloading because it can return the result to the requester.
But, helper $2$ is not success.
Therefore, the offloading success is closely related to $t_{ij}$, $t^c_{jk}$, and $t^s_{jk}$.
\begin{figure}[ht]
\centering
\includegraphics[scale=0.65]{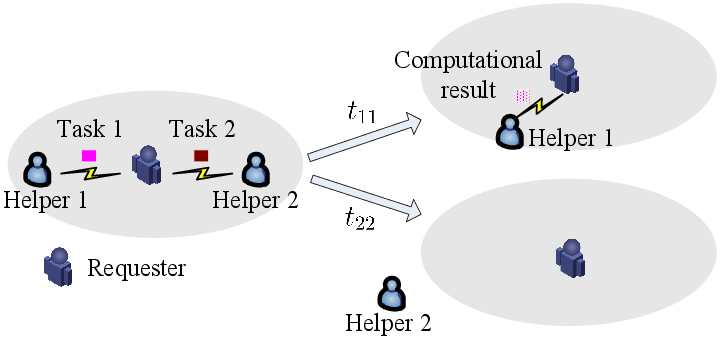}
\begin{center}
\caption{An example of the offloading process.}
\label{OffloadingProcess1}
\end{center}
\end{figure}

\begin{figure*}[ht]
\centering
\includegraphics[scale=0.95]{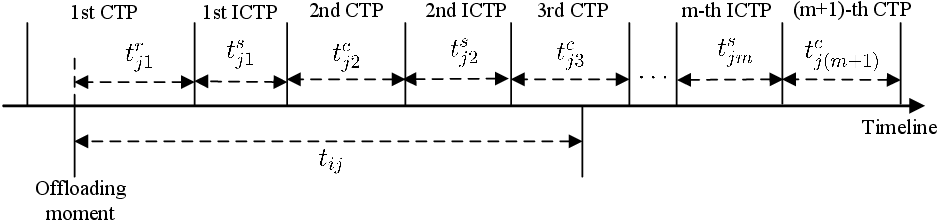}
\begin{center}
\caption{The timeline of the offloading process.}
\label{OffloadingProcess}
\end{center}
\end{figure*}
Fig. \ref{OffloadingProcess} describes the timeline of the offloading process for task $i$ and helper $j$, $i=1,2,\dots,R, ~j=1,2,\dots,H$.
The offloading is successful in the residual CTP only if
\begin{equation}
\begin{aligned}
t_{ij}<t^r_{j1}.
\end{aligned}
\end{equation}
In addition, the offloading is successful in the $k$-th CTP only if
\begin{equation}
\left\{
\begin{aligned}
{}& t^r_{j1}+t^s_{j1} < t_{ij} < t^r_{j1}+t^s_{j1}+ t^c_{j2},~k=2, \\
{}& t^r_{j1}+t^s_{j1}+\sum_{k'=2}^{k-1} (t^{s}_{jk'}+t^{c}_{jk'})<t_{ij}< \\
{}& ~~~t^c_{j0}+t^s_{j1}+\sum_{k'=2}^{k-1} (t^{s}_{jk'}+t^{c}_{jk'})+t^{c}_{jk},~k>2.
\end{aligned}
\right.
\end{equation}

The offloading success probability of task $i$ offloaded to helper $j$, denoted by $p_{ij}$, is expressed in (\ref{1111}) at the top of next page.
\begin{figure*}[htb]
\hrulefill
\normalsize
\begin{equation}
\begin{aligned}
p_{ij}
=& p(t^c_{j0}>t_{ij})+p(t^r_{j1}+t^s_{j1} < t_{ij} < t^r_{j1}+t^s_{j1}+ t^c_{j2})\\
&+\sum_{k=3}^{+\infty}p(t^r_{j1}+t^s_{j1}+\sum_{k'=2}^{k-1} (t^{c}_{jk'}+t^{s}_{jk'})<t_{ij}<t^c_{j0}+t^s_{j1}+\sum_{k'=2}^{k-1} (t^{c}_{jk'}+t^{s}_{jk'})+t^{c}_{jk})
\end{aligned}
\label{1111}
\end{equation}
\hrulefill
\vspace*{4pt}
\end{figure*}
Denote by $f^*_{j}(s) $, $w^*_{j}(s) $, and $h^*_{ij}(s)$, the Laplace transform of $f_j(t) $, $w_j(t) $, and $h_{ij}(t)$, respectively. Thus,
\begin{equation}
\left\{
\begin{aligned}
{}& f^*_j(s)=\int_{t=0}^{+\infty}f_j(t)e^{-st}dt,\\
{}& w^*_j(s)=\int_{t=0}^{+\infty}w_j(t)e^{-st}dt,\\
{}& h^*_{ij}(s)=\int_{t=0}^{+\infty}h_{ij}(t)e^{-st}dt.
\label{Lap}
\end{aligned}
\right.
\end{equation}
Denote by $f^*_{r_j}(s)$ the Laplace transform of $f_{r_j}(t)$. It follows that
\begin{equation}
\begin{aligned}
f^*_{r_j}(s)=\frac{\mu_j(1-f^*_j(s))}{s}.
\end{aligned}
\end{equation}


Applying Propositions 1 and 2 in Sect. II to (\ref{1111}), we derive (\ref{Resid}) at the top of next page.
\begin{figure*}[t]
\normalsize
\begin{equation}
\begin{aligned}
p_{ij}
=& p(t^c_{j0}>t_{ij})+p(t^r_{j1}+t^s_{j1} < t_{ij} < t^r_{j1}+t^s_{j1}+ t^c_{j2})\\
&+\sum_{k=3}^{+\infty}p(t^r_{j1}+t^s_{j1}+\sum_{k'=2}^{k-1} (t^{c}_{jk'}+t^{s}_{jk'})<t_{ij}<t^c_{j0}+t^s_{j1}+\sum_{k'=2}^{k-1} (t^{c}_{jk'}+t^{s}_{jk'})+t^{c}_{jk})\\
= &-\underset{q \in \Omega_f}\sum\underset{s=q}{\text{Res}}\frac{h^*_{ij}(s)}{s}f_{r_j}^*(-s)
+
\sum_{k=2}^{+\infty}
\underset{q \in \Omega_h}\sum\underset{s=q}{\text{Res}}\frac{f^*_{r_j}(s)f^*_j(s)^{k-2}w^*_j(s)^{k-1}(f^*_{j}(s)-1)}{s}h^*_{ij}(-s).
\end{aligned}
\label{Resid}
\end{equation}
\hrulefill
\vspace*{4pt}
\end{figure*}
In (\ref{Resid}), $\Omega_f$ and $\Omega_h$  represent the set of poles of $f_j^*(-s)$ and $h_{ij}^*(-s)$, respectively.
The offloading success probability of task $i$, denoted by $p_i$, is expressed as
\begin{equation}
\begin{aligned}
p_i = \sum_{j=1}^{H}p_{ij}x_{ij}.
\end{aligned}
\label{13}
\end{equation}


Thus, the average offloading success probability over all tasks, denoted by $p$,
is expressed as
\begin{equation}
\begin{aligned}
p = \frac{1}{R}\sum_{i=1}^{R}p_{i}.
\end{aligned}
\label{15}
\end{equation}

\section{Problem Modelling}
\subsection{Problem Formulation}
Our problem is to maximize the average offloading success probability over all tasks.
The task offloading maximization problem (TOMP) is expressed in (\ref{MATO}).
\begin{figure}[!h]
\begin{subequations}
\begin{alignat}{2}
&\max\limits_{\bm{x}}\quad p
 \label{F_e} \\
\text{s.t}. \quad
& (\ref{CapConst}),~(\ref{NumCons}),\nonumber\\
& x_{ij} \in \{0,1 \}, ~~i\in \mathcal{R}, j\in \mathcal{H}.
\end{alignat}
\label{MATO}
\end{subequations}
\end{figure}

\subsection{Complexity Analysis}
\begin{theorem}
\textbf{TOMP} is $\mathcal{NP}$-hard.
\end{theorem}

\begin{proof}
The proof is established by a reduction from the Knapsack problem that is NP-hard \cite{garey1979computers}.
Consider a Knapsack problem with a set of $N$ items and a knapsack capacity $W$.
Each item $i$, $i=1,2,\dots,N$, has a positive weight $w_i$ and a positive value $\theta_i$.
The Knapsack problem asks which items to be selected for the knapsack such that the total value of the selected items is maximized subject to the capacity of the knapsack.

We construct a TOMP reduction in the following.
We set $H=1$. That is, there is only one helper.
The amount of capacity of the helper, i.e., $E$, corresponds to $W$.
The size of task $i$, i.e., $l_i$, corresponds to $w_i$.
The offloading success probability of task $i$, i.e., $p_i$, corresponds to $\theta_i$.
From the above, solving the defined instance of TOMP will solve the Knapsack problem which is NP-hard. Hence the conclusion.
\end{proof}

\section{Algorithm Design}

\subsection{Uniform Scenario}
\subsubsection{Problem Analysis}
For the uniform scenario where $\mu_{j}=\mu$, $\gamma_{j}=\gamma$, and $E_j=E$, $ j\in \mathcal{H}$,
we propose a TSDP algorithm, as well as derive an upper bound of the global optimum for performance benchmarking.
In the uniform scenario, the helpers are identical in terms of performance.
Thus, it is unnecessary to identify which helper computes which task.
The offloading performance depends only on whether or not task $i$ is offloaded to the helpers.
The optimization variable $\bm{x}$ is re-defined as
 \begin{equation}
\begin{aligned}
\bm{x}=\{x_i, i\in \mathcal{R} \}.
\end{aligned}
\label{xi}
\end{equation}
 In (\ref{xi}), $x_i$ is one if and only if task $i$ is offloaded.
TOMP can then be reformulated as (\ref{MATOre1}).
\begin{figure}[!h]
\begin{subequations}
\begin{alignat}{2}
&\max\limits_{\bm{x}}\quad p
 \label{F_ere} \\
\text{s.t}. \quad
& \sum_{i=1}^{R}x_i l_i \leq E,~ \text{for any helper}, \label{F_bre} \\
& x_i \in \{0,1 \}, ~~i\in \mathcal{R}.
\end{alignat}
\label{MATOre1}
\end{subequations}
\end{figure}

In order to solve (\ref{MATOre1}), we first relax (\ref{F_bre}), and then derive a relaxation of (\ref{MATOre1}), which is expressed in (\ref{MATOre}).
\begin{figure}[!h]
\begin{subequations}
\begin{alignat}{2}
&\max\limits_{\bm{x}}\quad p
 \label{F_ere} \\
\text{s.t}. \quad
& \sum_{i=1}^{R} x_i l_i \leq E H, \label{F_brer} \\
& x_i \in \{0,1 \}, ~~i\in \mathcal{R}.
\end{alignat}
\label{MATOre}
\end{subequations}
\end{figure}

\begin{theorem}
The optimum of (\ref{MATOre}) is an upper bound of that of (\ref{MATOre1}).
\end{theorem}

\begin{proof}
Compared to (\ref{MATOre1}), (\ref{MATOre}) has larger solution space.
Therefore, the optimum of (\ref{MATOre}) is an upper bound of that of (\ref{MATOre1}).
\end{proof}
We remark that the upper bound can be used to benchmark any sub-optimal algorithm.

\subsubsection{Dynamic Programming Algorithm}
We develop a dynamic programming (DP) algorithm to derive the optimum of (\ref{MATOre}).
We assume that $E$ is integer.
Denote by $x^*_i$ the optimum of $x_i$, $i \in \mathcal{R}$.
Denote by $p^*(r,q)$ the offloading success probability of the optimum
with regard to the first $r$ tasks using an amount of capacity of $q$.
The recursive function in Lemma \ref{Lemma_Recursive} is used for the value of $p^*(r,q)$.

\begin{lemma}\label{Lemma_Recursive}
The value of $p^*(r,q)$ is derived by the recursive function in (\ref{RecursiveFunction1111}).
\end{lemma}

 \begin{equation}
p^*(r,q)=\left\{
\begin{aligned}
{}&\underset{x_r \in \{0,1 \}~~~~~~~~~~~~~~~~~~~~~~~~~~~~~~~~~~~}{\max~ \{\psi(x_r)+p^*(r-1,q- x_r l_r)\}},&& r>1,\\
{}& \underset{x_r \in \{0,1 \}~~~~~~~}{\max~ \{\psi(x_r)\}},&& r=1.
\end{aligned}
\right.
\label{RecursiveFunction1111}
\end{equation}
In (\ref{RecursiveFunction1111}), $\psi(x_r)$ is the offloading success probability depending on whether task $r$ is offloaded or not. Obviously, if it is not offloaded, i.e., $x_r=0$, $\psi(x_r)=0$.
If $x_r=1$, $\psi(x_r)$ can be derived by (\ref{Resid}).

\begin{proof}
We prove Lemma 2 by induction.
For $r=1$, the conclusion is obvious.
More specifically, if $q\ge l_r$, the optimum is $x^*_r=1$.
Otherwise, the optimum is $x^*_r=0$.
For $r>1$,
by (\ref{RecursiveFunction1111}),
 \begin{equation}
\begin{aligned}
 p^*(r+1,q)=\underset{x_{r+1} \in \{0,1 \}~~~~~~~~~~~~~~~~~~~~~~~~~~~~~~~~~~~~~~~}{\max~ \{\psi(x_{r+1})+p^*(r,q- x_{r+1} l_{r+1})\}}.
\end{aligned}
\end{equation}
For any $x_{r+1}$, $p^*(r,q- x_{r+1} l_{r+1})$ is the offloading success probability of the optimum with regard to the first $r$ tasks using an amount of capacity of $q-x_{r+1} l_{r+1}$, and $\psi(x_{r+1})$ is the offloading success probability whether task $r+1$ is offloaded or not, thus together leading to that $p^*(r+1,q)$ is the offloading success probability of the optimum.
Hence the conclusion.
\end{proof}
\begin{figure*}[ht]
\centering
\includegraphics[scale=0.8]{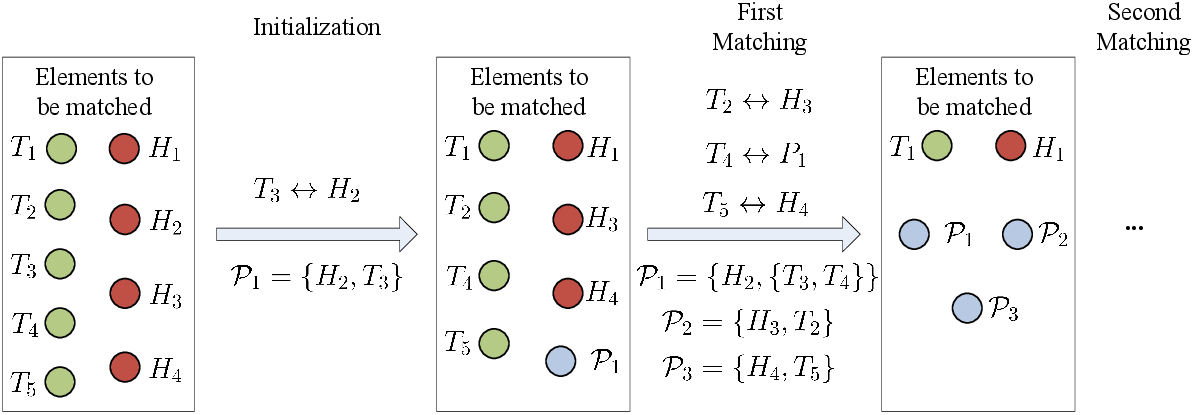}
\begin{center}
\caption{An example of the RMA algorithm.}
\label{ExampleRepeatedMatching}
\end{center}
\end{figure*}
\subsubsection{Algorithm Summary and Complexity Analysis}
Algorithm \ref{alg2} describes the proposed dynamic programming (DP) algorithm.
The input of the algorithm includes $\mu$, $\lambda$, $\delta$, $l$ and $E'$.
First, all the entries of $\bm{x}$ are initialized by Line \ref{ini}, i.e., $\bm{x} \leftarrow [0]_{1\times R}$.
Then, Lines 5 and 6 compute the optimum of $x_1$.
Finally, Lines 8 and 9 derive the optimum of $x_r$, $r=2,3,\dots,R$.
The computational complexity of Algorithm \ref{alg2} is of $O(RHE)$.
 \begin{algorithm}
\caption{DP algorithm}
\label{alg2}
\begin{algorithmic}[1]
\algsetup{linenosize=\tiny}
\small
\REQUIRE $\mathcal{R}$, $\bm{\mu}$, $\bm{\gamma}$, $\bm{\xi}$, $\bm{l}$, $H$, and $E$
\ENSURE $\bm{x}$
\STATE $\bm{x} \leftarrow [0]_{1\times R}$ \label{ini}
\FOR{$q=0$ : $EH$}
\FOR{$r=1$ : $R$}
\IF {$r=1$}
\IF{$q\ge l_r$}
\STATE $p^*(r,q) \leftarrow \psi(1) $
\STATE $ x^*_r \leftarrow 1 $
\ELSE
\STATE  $p^*(r,q) \leftarrow 0$
\STATE  $ x^*_r \leftarrow 0$
\ENDIF
\ELSE
\STATE $p^*(r,q) \leftarrow \underset{x_r \in \{0,1 \}~~~~~~~~~~~~~~~~~~~~~~~~~~~~~~~~~~~~}{\max~ \{\psi(x_r)+p^*(r-1,q- x_r l_r)\}} $
\STATE  $ x^*_r \leftarrow \underset{x_r \in \{0,1 \}~~~~~~~~~~~~~~~~~~~~~~~~~~~~~~~~~~~~~~~~~~~~~}{\argmax\{\psi(x_r)+p^*(r-1,q- x_r l_r)\}~~~~~}$
\ENDIF
\ENDFOR
\ENDFOR
\RETURN $\bm{x}$
\end{algorithmic}
\end{algorithm}

\begin{theorem}
Algorithm \ref{alg2} derives the optimum of the problem in (\ref{MATOre}).
\end{theorem}

\begin{proof}
The optimality is concluded by Lemma \ref{Lemma_Recursive}.
\end{proof}

\subsubsection{Obtaining Feasible Solution}
Algorithm 1 is from a relaxation point of view.
There is no guarantee that the solutions by the algorithm are feasible.
Thus, we propose TSDP to obtain a feasible solution of (\ref{MATOre1}).
Denote by $\mathcal{R}'$ the set of unassigned tasks.
A general description of the algorithm is the following.
We take any helper and the capacity is $E$.
Initially, $\bm{x}$ is set to be zero.
The algorithm starts with the first helper and optimizes the helper by the DP shown in Algorithm 1.
Once the current helper is optimized, $\mathcal{R}'$ will be updated accordingly.
We repeat the above process for the remaining $H-1$ identical helpers.
The overall progress of the TSDP algorithm is shown in Algorithm 2.
The complexity of the algorithm is $O(RHE)$.
 \begin{algorithm}
\caption{TSDP algorithm}
\begin{algorithmic}[1]
\algsetup{linenosize=\tiny}
\small
\STATE $\bm{x}^* \leftarrow [0]_{1\times R}$ \label{ini}
\FOR{$j=1:H$}
\STATE Derive the selected tasks by applying the DP algorithm.
\STATE $\bm{x}^*$ $\leftarrow$ the selected tasks.
\STATE Delete the selected tasks in $\mathcal{R}'$.
\ENDFOR
\RETURN $\bm{x}^*$
\end{algorithmic}
\end{algorithm}


\subsection{General Scenarios}
For general scenarios of TOMP, we develop a RMA algorithm in which we solve a series of matching problems, and input matrix is updated in every iteration, until the solution can not be improved.
\subsubsection{Matching Problem}
The basic matching problem is the following.
Given a set $\mathcal{B}$ with $d$ elements, $\mathcal{B}=\{b_1,b_2,\dots,b_d\}$.
A perfect matching on $\mathcal{B}$ is defined as a matching of elements in $\mathcal{B}$ such that each $b_i \in \mathcal{B}$ is matched with exactly one $b_j \in \mathcal{B}$, $i \neq j$.
Denote by $v_{ij}$ the value matching $b_i$ with $b_j$, $i,j=1,2,\dots,d,~ i\neq j$, where $v_{ij}=v_{ji}$.
Denote by $z_{ij}$ a binary variable, which is one if and only if $b_i$ is matched with $b_j$, otherwise zero.
The maximum value perfect matching problem is to find a perfect matching on $\mathcal{B}$ such that the sum of the values of the pairs of matched elements is maximized.
It can be formulated in (\ref{MP}).
\begin{figure}[!h]
\begin{subequations}
\begin{alignat}{2}
&\max\limits_{\bm{z}\in \{0,1\}}\quad \sum_{i=1}^{d}\sum_{j=1,j\neq i}^{d}v_{ij}z_{ij}
 \label{F_ere} \\
\text{s.t}. \quad
& \sum_{i=1,i\neq j}^{d}z_{ij}=1,~j=1,2,\dots,d, \\
& \sum_{j=1,j\neq i}^{d}z_{ij}=1,~i=1,2,\dots,d.
\end{alignat}
\label{MP}
\end{subequations}
\end{figure}

It is well known that (\ref{MP}) can be solved by Kuhn-Munkres (KM) algorithm \cite{Munkres1957}.
The complexity of KM is $O(d^3)$.
\subsubsection{RMA algorithm}
In order to design the RMA algorithm to solve the problem in (\ref{MATO}),
we need to reformulate the problem.
In our problem, $\mathcal{R}=\{1,2,\dots,R\}$ and $\mathcal{H}=\{1,2,\dots,H\}$ denote the sets of all tasks and helpers, respectively.
Denote by $\mathcal{K}$ a set, $\mathcal{K}=\mathcal{H} \times \mathcal{D}$, where $\mathcal{D}$ represents the set of all nonempty subsets of $\mathcal{R}$, $\mathcal{D}=\{\mathcal{G}_1,\mathcal{G}_2,\dots\}$.
If $\sum_{i\in \mathcal{G}_1}l_i \le E_j$, $(j,\mathcal{G}_1)\in \mathcal{K}$ is feasible.
Denote by $\mathcal{F}$ the set of all feasible elements of $\mathcal{K}$.
A packing, denoted by $\mathcal{P}$, is defined as a subset of $\mathcal{F}$ satisfying the property that
\begin{equation}
\begin{aligned}
(j_1,\mathcal{G}_1),(j_2,\mathcal{G}_2) \in \mathcal{P}\Rightarrow \mathcal{G}_1 \cap \mathcal{G}_2 = \emptyset ~\text{and}~ j_1 \neq j_2.
\end{aligned}
\end{equation}

Given a packing $\mathcal{P}$, we define three sets, i.e., $\mathcal{L}_1=\{j|(j,\mathcal{D})\notin \mathcal{P}\}$, $\mathcal{L}_2=\cup_{(j,\mathcal{D})\notin \mathcal{P}} \mathcal{G}$, and $\mathcal{L}_3=\mathcal{P}$.
The set $\mathcal{L}_1$ consists of all helpers that are not used, $\mathcal{L}_2$ consists of all tasks that are not allocated to a helper, and $\mathcal{L}_3$ denotes the set consisting of all used helpers with their allocated tasks.
Denote by $m_1$, $m_2$, and $m_3$ the cardinalities of $\mathcal{L}_1$, $\mathcal{L}_2$, and $\mathcal{L}_3$, respectively.
The offloading success probability of the packing is expressed as
\begin{equation}
\begin{aligned}
\delta m_2 +\sum_{(j,i)|(j,\mathcal{D})\in \mathcal{P},i\in \mathcal{D}} p_{ij},
\label{Packi}
\end{aligned}
\end{equation}
where $\delta$ is some small negative number, which is a penalty parameter so as to decrease the number of tasks that are not allocated.

RMA is to match the elements of $\mathcal{L}_1$, $\mathcal{L}_2$, and $\mathcal{L}_3$ with each other, thus generating new sets $\mathcal{L}'_1$, $\mathcal{L}'_2$, and $\mathcal{L}'_3$, such that the offloading success probability of the new packing is improved over $\mathcal{L}'_3$.
Fig. \ref{ExampleRepeatedMatching} gives an example of RMA in which the matched elements include five tasks and four helpers, i.e., $T_1,\dots,T_5$ and $H_1,\dots,H_4$.
We initialize that $T_3$ is matched with $H_2$, i.e., $\mathcal{P}_1=\{H_2,T_3\}$.
After the initialization, $\mathcal{L}_1=\{H_1,H_3,H_4\}$, $\mathcal{L}_2=\{T_1,T_2,T_4,T_5\}$, and $\mathcal{L}_3=\{\mathcal{P}_1\}$.
In the first matching, suppose that $T_2$ is matched with $H_3$, $T_4$ is matched with $\mathcal{P}_1$, and $T_5$ is matched with $H_4$, thus generating two new packing $\mathcal{P}_2=\{H_3,T_2\}$ and $\mathcal{P}_3=\{H_4,T_5\}$ and updating $\mathcal{P}_1=\{H_2,\{T_3,T_4\}\}$.
After the first matching, $\mathcal{L}'_1=\{H_1\}$, $\mathcal{L}'_2=\{T_1\}$, and
$\mathcal{L}'_3=\{\mathcal{P}_1,\mathcal{P}_2,\mathcal{P}_3\}$.
Then, we continue to the second matching, until the offloading success probability of the new packing cannot be improved over $\mathcal{L}'_3$.

Given a packing $\mathcal{P}$, a key step is to derive value coefficients, i.e., $v_{ij}$ in (\ref{MP}).
The value matrix's dimension is $(m_1+m_2+m_3)\times (m_1+m_2+m_3)$.
As we need to match the elements of $\mathcal{L}_1$, $\mathcal{L}_2$, and $\mathcal{L}_3$ with each other, the value matrix includes nine submatrices referred to as block 1,$\dots$,block 9, expressed as
\begin{equation}
\begin{aligned}
\bm{v}={}&
\begin{bmatrix}
[\mathcal{L}_1 \leftrightarrow \mathcal{L}_1] & [\mathcal{L}_1 \leftrightarrow \mathcal{L}_2] & [\mathcal{L}_1 \leftrightarrow \mathcal{L}_3]\\
[\mathcal{L}_2 \leftrightarrow \mathcal{L}_1] & [\mathcal{L}_2 \leftrightarrow \mathcal{L}_2] & [\mathcal{L}_2 \leftrightarrow \mathcal{L}_3]\\
[\mathcal{L}_3 \leftrightarrow \mathcal{L}_1] & [\mathcal{L}_3 \leftrightarrow \mathcal{L}_2] & [\mathcal{L}_3 \leftrightarrow \mathcal{L}_3]
\end{bmatrix}
\\
={}&
\begin{bmatrix}
[1] & [2] & [3]\\
[4] & [5] & [6]\\
[7] & [8] & [9]
\end{bmatrix}.
\end{aligned}
\label{Vmatrix}
\end{equation}
In (\ref{Vmatrix}), $\leftrightarrow$ is used as a notation for matching.
As $\bm{v}$ is a symmetric matrix, we need to compute only blocks 1, 4, 5, 7, 8, and 9.

\emph{Blocks 1, 4, and 5}: In block 1, a matching of any two unused helpers is not feasible.
Thus, the value can be set to $-\infty$.
In block 4, a matching of an unused helper with an unassigned task is feasible if the size of the task does not exceed the communication capacity with the helper.
The value is the offloading success probability that the task is assigned to the helper.
In block 5, a matching of any two unassigned tasks is not feasible.
Thus, the value can be set to $-\infty$.

\emph{Block 9}: To compute the value of a matching of any two elements in packing $\mathcal{P}$, e.g., $(j_1,\mathcal{G}_1)$ and $(j_2,\mathcal{G}_2)$, we divide the matching into three different cases.
The first two cases are when all tasks in $\mathcal{G}_1$ and $\mathcal{G}_2$ are allocated to one of the two helpers, i.e., $j_1$ or $j_2$.
It is easy to compute the value of the two cases by comparing the total size of the tasks with the capacity of each helper.
The third corresponds to the case when some tasks become reallocated from one helper to the other.
For this case, we need to compute the optimal reallocation of tasks.
This can be accomplished by solving an integer programming problem.
Denote by $w_i$ a binary variable, which is one if and only if task $i$, $i\in \mathcal{G}_1$, is reallocated to helper $j_2$.
Denote by $y_i$ a binary variable, which is one if and only if task $i$, $i\in \mathcal{G}_2$, is reallocated to helper $j_1$.
Denote by $\xi_w$ and $\xi_y$ the spare capacities at helpers $j_1$ and $j_2$, respectively.
Thus, the optimal task reallocation problem can be formulated as
\begin{figure}[!h]
\begin{subequations}
\begin{alignat}{2}
&\max\limits_{\bm{w},\bm{y}}\quad \sum_{i\in \mathcal{G}_1} (p_{ij_2}-p_{ij_1}) w_i+\sum_{i\in \mathcal{G}_2} (p_{ij_1}-p_{ij_2}) y_i
 \label{ObjSwap} \\
\text{s.t}. \quad
& \sum_{i\in \mathcal{G}_1} l_i w_i-\sum_{i\in \mathcal{G}_2} l_i y_i \le \xi_y, \label{SwapC1} \\
& -\sum_{i\in \mathcal{G}_1} l_i w_i+\sum_{i\in \mathcal{G}_2} l_i y_i \le \xi_w, \label{SwapC2}\\
& w_i, y_i \in \{0,1\}.
\end{alignat}
\label{SwapOpt}
\end{subequations}
\end{figure}

In (\ref{SwapOpt}), the objective function (\ref{ObjSwap}) is to maximize the offloading success probability by the reallocation of tasks.
Constraints (\ref{SwapC1}) and (\ref{SwapC2}) denote that the reallocation of tasks cannot exceed the spare capacities.
We propose a fast reallocation algorithm to solve the problem.
Algorithm 3 describes the process of the task reallocation algorithm.
In Lines 1-2, we use a vector to store the difference value results if each task is reallocated.
In Line 3, we sort the vector elements in descending order.
In Line 4, for the positive elements, we reallocate the corresponding tasks from a helper to another helper one by one if the remaining capacity of the latter is greater than the size of the task, and update the capacities of the two helpers.
By comparing the values of three different cases and selecting the best, we derive the matching value of Block 9. The complexity of Algorithm 3 is $O(R^2)$.
 \begin{algorithm}
\caption{Reallocation algorithm}
\begin{algorithmic}[1]
\algsetup{linenosize=\tiny}
\small

\STATE For each $i$, $i \in \mathcal{G}_1$, compute $\delta_{ij_1}=p_{ij_1}-p_{ij_2}$, and the result is stored in vector $\Delta$;
\STATE For each $i$, $i \in \mathcal{G}_2$, compute $\delta_{ij_2}=p_{ij_1}-p_{ij_1}$, and the result is stored in vector $\Delta$;
\STATE Sort $\Delta$ in descending order.
\STATE For the elements greater than zero, reallocating the corresponding task from a helper to another helper one by one if the spare capacity of the latter is greater than the size of the task, and update the capacities of the two helpers.
\end{algorithmic}
\end{algorithm}

\emph{Blocks 7 and 8}: In block 7, in order to compute the value of matching an element $(j_1,\mathcal{G}_1)\in \mathcal{L}_3$ with an unused helper $j_2 \in \mathcal{L}_1$, we use an auxiliary empty set $\mathcal{G}_2$ and $(j_2,\mathcal{G}_2)\in \mathcal{L}_1$.
In block 8, in order to compute the value of matching an element $(j_1,\mathcal{G}_1)\in \mathcal{L}_3$ with unassigned tasks $\mathcal{G}_2 \in \mathcal{L}_1$, we use an auxiliary helper $j_2$.
The capacity of $j_2$ is zero.
Thus, computing the values of blocks 7 and 8 are similar to that of block 9.

The idea of the RMA algorithm is to explore different packing solutions.
The algorithm is started via considering any feasible packing, i.e., initializing $\mathcal{L}_1$, $\mathcal{L}_2$, $\mathcal{L}_3$, and $MaxI$, where $MaxI$ denotes the maximum number of iterations.
It ends until the result of (\ref{Packi}) cannot be improved.
Algorithm 4 summarizes the overall process of RMA.

 \begin{algorithm}
\caption{RMA algorithm}
\begin{algorithmic}[1]
\algsetup{linenosize=\tiny}
\small

\STATE \textbf{Initialize} $Opt=-\infty$, $label=1$, $CountI=1$, $\mathcal{L}_1$, $\mathcal{L}_2$, $\mathcal{L}_3$, and $MaxI$;
\WHILE{$label=1$ and $CountI \le MaxI$}
\STATE Compute $\bm{v}$;
\STATE Solve (\ref{MP}) and obtain the result $p'$ by (\ref{Packi});
\IF{$p'> Opt$}
\STATE $Opt=p'$;
\STATE Update $\mathcal{L}_1$, $\mathcal{L}_2$, and $\mathcal{L}_3$;
\ELSE
\STATE $label=0$;
\ENDIF
\STATE $CountI=CountI+1$;
\ENDWHILE
\end{algorithmic}
\end{algorithm}

\section{Performance Evaluations}
We have proposed the TSDP and RMA algorithms for the uniform and general scenarios of TOMP, respectively.
We evaluate the effectivenesses of the two algorithms by comparing them to the following algorithms.
\begin{itemize}

\item
Monte Carlo search algorithm (MCSA): The tasks are processed one by one.
Each task is randomly assigned to a helper if the helper's spare capacity can accommodate task.
Then, we derive the average offloading success probability of all the tasks.
The above process is iterated 10000 times.
The maximization offloading success probability is selected as the solution of MCSA.
\item
Greedy algorithm (GA): The tasks are processed one by one.
For any task $i$, all the helpers are sorted in descending order by ($\frac{1}{\mu_j}+\frac{1}{\gamma_j}+\frac{1}{\xi_{ij}}$).
The task is assigned to the first helper if the helper's spare capacity can accommodate task.
Otherwise, the task is assigned to the latter helpers.
\item
Upper bound: In uniform scenarios, an upper bound is derived via (\ref{MATOre}).
\end{itemize}

The sizes of tasks are randomly selected in $[1, l_{max}]$.
The communication capacities to helpers are randomly selected in $[1, E_{max}]$.

\subsection{Validation of Average Offloading Success Probability}
Given the distribution of $t^c_{jk}$, $t^s_{jk}$, and $t_{ij}$, the closed-form expression of the offloading success probability can be derived.
In performance evaluations, we assume that $t^c_{jk}$ and $t^s_{jk}$ follow an exponential distribution, respectively.
That is
\begin{equation}
\left\{
\begin{aligned}
{}& f_j(t)=\mu_je^{-\mu_j t},\\
{}& w_j(t)=\gamma_je^{-\gamma_j t}.
\end{aligned}
\right.
\end{equation}
This assumption is based on three aspects.
First, the works in \cite{Deng2020Model,Cost2018Deng,Wang2018twc,Deng2019Model} have used exponential distribution to describe the nodes' mobility.
Second, the work in \cite{Zhu2010Recognizing} analyzes the real-world nodes' mobility, and finds that the tail behavior of ICTP can be characterized by an exponential distribution.
Finally, an exponential distribution can characterize at least $80\%$ of CTP distributions by investigating realistic mobility traces \cite{YLi2013Revealing}.
We assume that $h_{ij}(t)$ follows an Erlang distribution,
\begin{equation}
\begin{aligned}
h_{ij}(t)=\frac{\xi_{ij}(\xi_{ij} t)^{n_h-1}}{(n_h-1)!}e^{-\xi_{ij} t}.
\end{aligned}
\end{equation}
Thus, their Laplace transforms are expressed as
\begin{equation}
\left\{
\begin{aligned}
{}& f^*_j(s)= \frac{\mu_j}{\mu_j+s},\\
{}& w^*_{j}(s)=\frac{\gamma_{j}}{\gamma_{j}+s},\\
{}& h^*_{ij}(s)=(\frac{\xi_{ij}}{\xi_{ij}+s})^{n_h}.
\end{aligned}
\label{ghs}
\right.
\end{equation}

Plugging (\ref{ghs}) into (\ref{Resid}), we derive the expression of $p_{ij}$.
 For example, when $n_h=1$, $p_{ij}$ is expressed as
\begin{equation}
\begin{aligned}
p_{ij}= \frac{\xi_{ij}+\gamma_j}{\mu_j+\xi_{ij}+\gamma_j}.
\end{aligned}
\label{one}
\end{equation}
When $n_h=2$, $p_{ij}$ is expressed as
\begin{equation}
\begin{aligned}
p_{ij}= &(\frac{\xi_{ij}}{\mu_j+\xi_{ij}})^2\frac{1}{1-q}+(\frac{\xi_{ij}}{\mu_j+\xi_{ij}})^2\frac{q}{(1-q)^2}\\
&+\frac{\xi^2_{ij}}{(\mu_j+\xi_{ij})(\gamma_j+\xi_{ij})}\frac{q}{(1-q)^2},
\end{aligned}
\label{two}
\end{equation}
where
\begin{equation}
\begin{aligned}
q=\frac{\mu_j\gamma_j}{(\mu_j+\xi_{ij})(\gamma_j+\xi_{ij})}.
\end{aligned}
\end{equation}

\begin{table}[h]
		\centering
		\begin{tabular}{|c|c|c|c|}\hline
			Parameters setting&Closed-form& Simulation & Gap\\\hline
			$n_h=1$&0.7873&0.7887&0.18$\%$\\\hline
            $n_h=2$&0.7208&0.7170&0.53$\%$\\\hline
		\end{tabular}
		\caption{Validation of $p_{ij}$.}
		\label{NumericalSimulation}
	\end{table}
We validate the accuracy of the closed form of $p_{ij}$ by a Monte Carlo
simulation, $i\in \mathcal{R}, j\in \mathcal{H}$.
For any $p_{ij}$, in order to derive a performance metric, we simulate a certain number of offloading process, denoted by $N_{off}$.
Here, $N_{off}=50000$.
In each simulation process, the communication time period $t^c_{kj}$, the inter-communication time period $t^s_{kj}$, and the processing cycles $t_{ij}$ are randomly generated via their corresponding distribution models, respectively.
Then, we count the total number of success offloading.
We derive the results of $p_{ij}$ by (\ref{one}) and (\ref{two}).
Table I compares the $p_{ij}$ and simulation results of $p_{ij}$ when $n_h=1$ and $n_h=2$.
It can observed that the gap is less than $0.55\%$, validating the accuracy of the closed-form expression.

\subsection{The Uniform Scenario}

Figs. \ref{ImpactH} and \ref{ImpactR} evaluate the impact of $H$ and $R$, respectively.
In the two figures, we use the result of TSDP to initialize RMA.
Overall, the average offloading success probability linearly increases with respect to $H$, and decreases with respect to $R$.
This is expected, because more helpers will help the requester offload its tasks.
Conversely, more tasks compete the fixed computation resources with the increase of $R$, and the average amount of offloading decreases.
For the TSDP algorithm, it is close to the bound no matter $R$ and $H$ increase or not.
Moreover, its solution is equal to the upper bound in some instances, indicating that TSDP obtains the optimum.
For the RMA algorithm, its performance is the same as TSDP, which means that given a good initialization, it is hard to improve the performance of RMA by solving a series of matching problems.
Finally, the TSDP and RMA algorithms outperform the MCSA algorithm.

\begin{figure}[ht]
\centering
\includegraphics[scale=0.6]{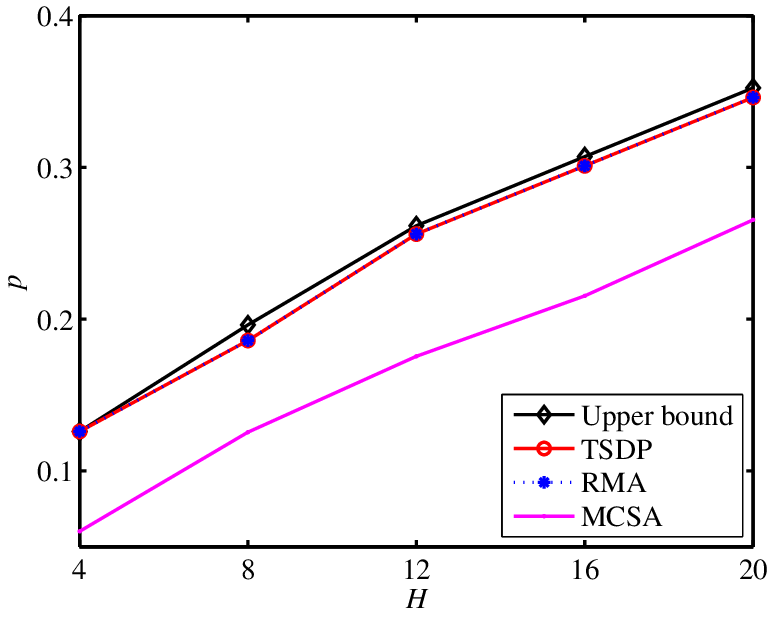}
\begin{center}
\caption{Impact of $H$ in the uniform scenario.}
\label{ImpactH}
\end{center}
\end{figure}
\begin{figure}[ht]
\centering
\includegraphics[scale=0.6]{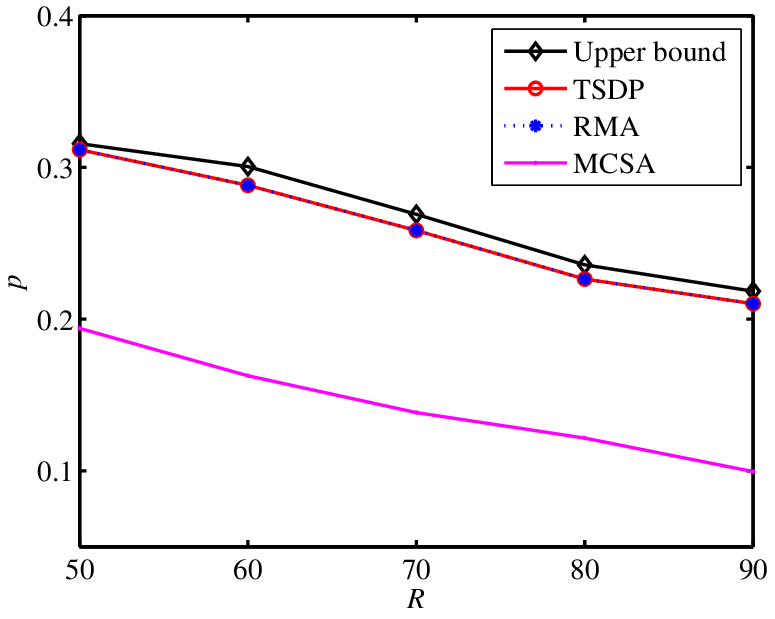}
\begin{center}
\caption{Impact of $R$ in the uniform scenario.}
\label{ImpactR}
\end{center}
\end{figure}

\subsection{General Scenarios}
\begin{figure}[ht]
\centering
\includegraphics[scale=0.6]{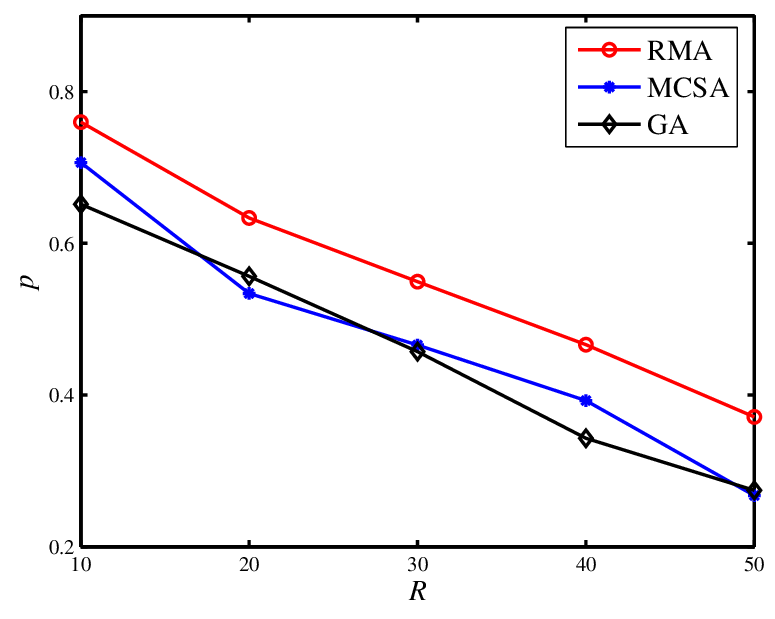}
\begin{center}
\caption{Impact of $R$ in general scenarios.}
\label{ImpactRGene}
\end{center}
\end{figure}
\begin{figure}[ht]
\centering
\includegraphics[scale=0.6]{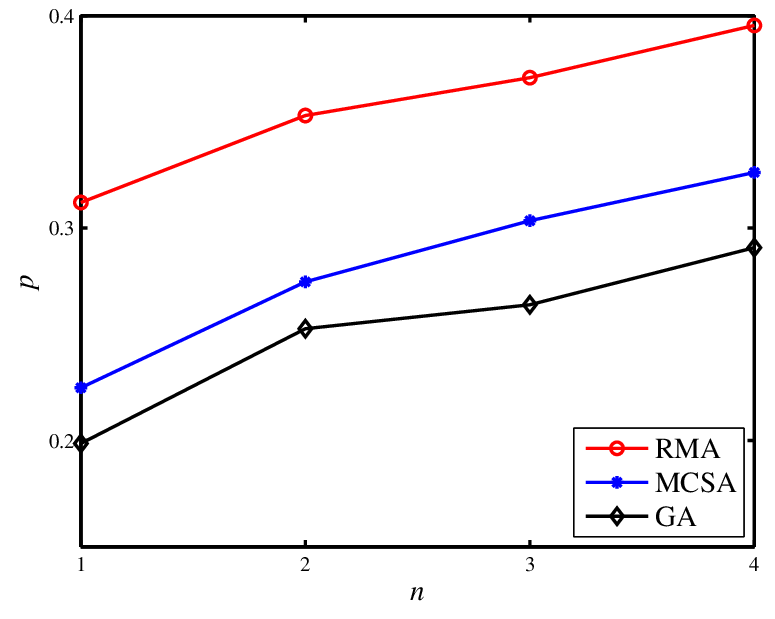}
\begin{center}
\caption{Impact of $\mu$ where $\mu$ folows the distribution of $\Gamma(4.43/n, 1/1088)$ in general scenarios.}
\label{ImpactMu}
\end{center}
\end{figure}
\begin{figure}[ht]
\centering
\includegraphics[scale=0.6]{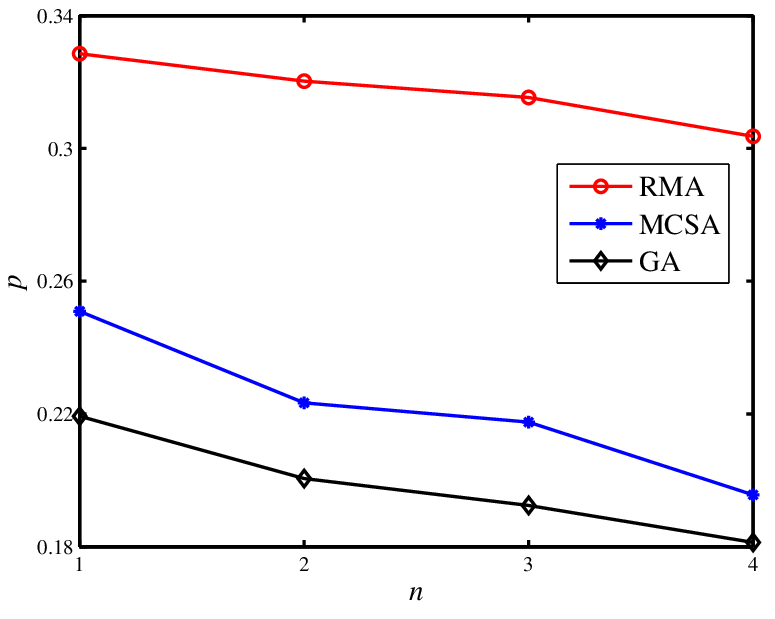}
\begin{center}
\caption{Impact of $\gamma$ where $\gamma$ folows the distribution of $\Gamma(4.43/n, 1/1088)$ in general scenarios.}
\label{ImpactGamma}
\end{center}
\end{figure}
\begin{figure}[ht]
\centering
\includegraphics[scale=0.6]{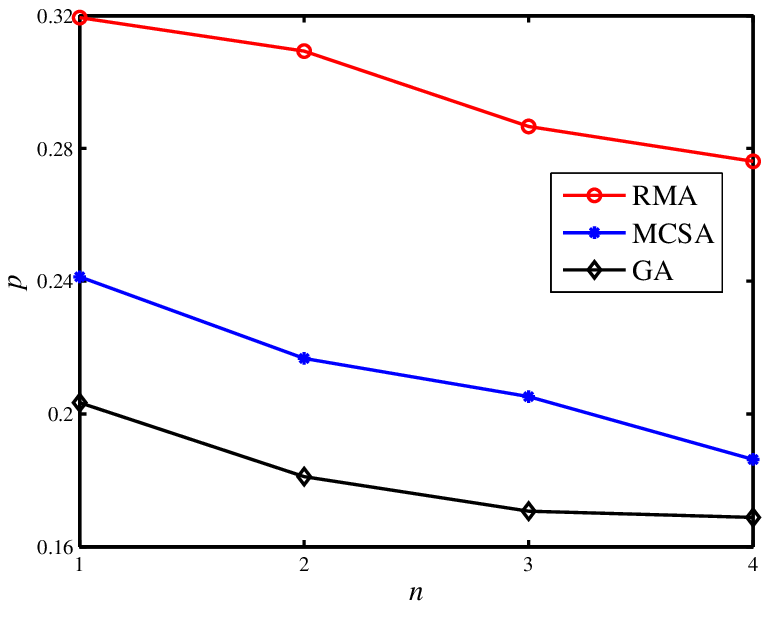}
\begin{center}
\caption{Impact of $\xi$ where $\xi$ folows the distribution of $\Gamma(4.43/n, 1/1088)$ in general scenarios.}
\label{ImpactXi}
\end{center}
\end{figure}

Fig. \ref{ImpactRGene} evaluates the impact of $R$ in general scenarios of TOMP.
It can be observed that RMA outperforms both MCSA and GA.
When $R=10$, RMA outperforms MCSA by $7.5\%$ and GA by $16.6\%$.
When $R=50$, RMA outperforms MCSA by $38.9\%$ and GA by $35.2\%$.
This phenomenon manifests that RMA is suitable for large-scale scenarios.
Fig. \ref{ImpactMu}-\ref{ImpactXi} evaluate the impact of $\bm{\mu}$, $\bm{\gamma}$, and $\bm{\xi}$.
In these figures, $\bm{\mu}$, $\bm{\gamma}$, and $\bm{\xi}$ are generated by a Gamma
distribution $\Gamma(4.43/n, 1/1088)$, where $n$ is a constant.
The following insights are obtained.
First, in Fig. \ref{ImpactMu}, when $n$ increases, the average offloading success probability increases.
For example, the probability increases from $31.2\%$ to $39.54\%$ for RMA.
This is because increasing $n$ causes that $\bm{\mu}$ decreases, such that the average communication time increases.
The helpers have more opportunity to return the computation results.
Second, in Fig. \ref{ImpactGamma}, when $n$ increases, the average offloading success probability decreases.
This is because increasing $n$ causes that $\bm{\gamma}$ and $\bm{\xi}$ decrease, such that the average inter-communication time increases.
Thus, the helpers are disconnected with the requester with high probability.
In Fig. \ref{ImpactXi}, when $n$ increases, the average offloading success probability decreases.
This is because increasing $n$ causes that $\bm{\xi}$ decreases, such that the average processing time increases.
The completion time of tasks are more likely fall in the inter-communication time.
Finally, RMA obviously outperforms MCSA and GA.

\section{Conclusions}
In this paper, we have investigated task offloading in MEC systems with uncertain computation
processing cycles and  intermittent communications.
First, we have derived a closed-form expression of the average offloading success probability of tasks, and formulated a task offloading maximization problem.
Then, we have proven that the problem is NP-hard.
For problem solving, we have proposed a fast and effective TSDP algorithm for the uniform scenario.
By solving the uniform case, we derived an upper bound enabling to benchmark.
For general scenarios, we have proposed a scalable RMA algorithm.
Finally, in performance evaluation, we have validated the accuracy of the
closed-form expression of the offloading success probability by Monte
Carlo simulation.
The gaps between the closed-form expression and simulation results are less than $0.55\%$, manifesting the accuracy of the closed-form expression.
In addition, we evaluate the performance of the
proposed algorithms by comparing them to other algorithms.
For the uniform scenario, TSDP is close to the upper bound as the gap to the upper bound does not
exceed $4\%$.
For general cases, RMA outperforms other algorithms.

\bibliographystyle{IEEEtran}
\bibliography{IEEEabrv,ForIEEEBib}

\end{document}